\documentclass[conference]{IEEEtran}

\usepackage{cite}      

\usepackage[usenames,dvipsnames,svgnames,table]{xcolor}
\usepackage{graphicx}  

\usepackage{psfrag}    

\usepackage{url}       
\usepackage{tikz}

\usepackage{amsmath}   
\usepackage{amssymb}

\usepackage{algorithm}
\usepackage{algpseudocode}
\usepackage{pifont}
\newcommand{\shorten}[1]{}
\newtheorem{proposition}{Proposition}
\newtheorem{theorem}{Theorem}

\newcommand{\signed}%
    {{\unskip\nobreak\hfill\penalty50
      \hskip2em\hbox{}\nobreak\hfil $\blacksquare$
      \parfillskip=0pt \finalhyphendemerits=0 \par}}
\newenvironment{proof}[1]
    {
    \bf{Proof:}\rm{\noindent{#1 }}\ignorespaces
    }
    {\signed\addvspace\medskipamount}

\begin{document}

\title{General Sub-packetized Access-Optimal Regenerating Codes}

%
\author{\authorblockN{
Katina Kralevska, Danilo Gligoroski, Harald {\O}verby}

}


\maketitle

\begin{abstract}
This paper presents a novel construction of $(n,k,d=n-1)$ access-optimal regenerating codes for an arbitrary sub-packetization level $\alpha$ for exact repair of any systematic node. We refer to these codes as general sub-packetized because we provide an algorithm for constructing codes for any $\alpha$ less than or equal to $r^{\lceil \frac{k}{r} \rceil}$ where $\frac{k}{r}$ is not necessarily an integer. This leads to a flexible construction of codes for different code rates compared to existing approaches. We derive the lower and the upper bound of the repair bandwidth. The repair bandwidth depends on the code parameters and $\alpha$. The repair process of a failed systematic node is linear and highly parallelized, which means that a set of $\lceil \frac{\alpha}{r} \rceil$ symbols is independently repaired first and used along with the accessed data from other nodes to recover the remaining symbols.


\end{abstract}


%
\IEEEpeerreviewmaketitle

\section{Introduction} \label{intro}

Erasure coding is becoming an attractive technique for data protection since it offers the same level of reliability with significantly less storage overhead compared to replication \cite{Weatherspoon:2002:ECV:646334.687814}.
Apart from the reliability and the storage overhead, there are other desirable features in a distributed storage system such as low repair bandwidth and access-optimality. Repair bandwidth is the amount of transferred data during a repair process. Access-optimality is achieved when the amount of accessed and transferred data during the repair process is equal. 


Dimakis et al. introduced regenerating codes that significantly reduce the repair bandwidth \cite{5550492}.
Under an $(n, k, d)$ regenerating code, a file of $M$ symbols from a finite field $\mathbf{F}_q$ is divided into $k$ fragments, each of size $\alpha=\frac{M}{k}$ symbols, which are further encoded into $n$ fragments using an $(n, k)$ MDS (Maximum-Distance Separable) code.
The parameter $\alpha$, termed as a sub-packetization level of the code, represents the minimum dimension over which operations are performed.
The data from a failed node is recovered by transferring $\beta$ symbols from each $d$ non-failed nodes. Thus, the repair bandwidth $\gamma$ is equal to $d\beta$ where $\alpha \leq d\beta \ll M$.
Dimakis et al. \cite{5550492} showed the existence of minimum storage regenerating (MSR) codes that attain the minimum storage point of the optimal tradeoff curve between the storage and the repair bandwidth, i.e.,
\begin{equation}
(\alpha_{MSR}, \gamma_{MSR}^{min})=(\frac{M}{k}, \frac{M}{k} \frac{n-1}{n-k}).
\label{optimalBW}
\end{equation}
The repair bandwidth is minimized when all $d=n-1$ non-failed nodes transmit a fraction of ${1}/{r}$ of the stored data.



Several exact repair MSR codes, which are characterized by the repaired data being exactly the same as the lost data, have been suggested.
Tamo et al. proposed optimal MSR codes for $\alpha$ equal to $r^k$ known as zigzag codes \cite{6352912}.
Furthermore, they showed that $\alpha$ of access-optimal MSR codes for repair of any systematic node is $r^\frac{k}{r}$ \cite{6283042}. The codes presented in \cite{6190343} and \cite{7084873} meet this condition.
An essential condition for the code construction in \cite{7084873} is that $m=\frac{k}{r}$ has to be an integer $m \geq 1$ where $k$ is set to $rm$ and $\alpha$ to $r^m$.
Wang et al. constructed codes that optimally repair any systematic or parity node for $\alpha$ equal to $r^{k+1}$ \cite{6120327}.
High-rate MSR codes with a polynomial sub-packetization level are proposed in \cite{7282816}.
However, our work focuses only on code constructions for optimal repair of any systematic node.

MSR codes are optimal in terms of storage, reliability and repair bandwidth, but not I/O.
Implementing MSR codes with a sub-packetization level of $r^\frac{k}{r}$ may not be practical for storage systems that serve applications with a large number of user requests or perform intensive computations.
Thus, having an algorithm for constructing MSR codes for any combination of $n$, $k$ and $\alpha$ that are simultaneously optimal in terms of storage, reliability, repair bandwidth and I/O is an important problem that is solved in this work.

\textbf{Our Contribution:}
This paper presents a novel construction of $(n, k, d=n-1)$ access-optimal regenerating codes for an arbitrary sub-packetization level for exact repair of any systematic node.
The codes have the following properties: 1. MDS; 2. Systematic; 3. Flexible sub-packetization level; 4. Minimum repair bandwidth for every $\alpha$ including the lower bound (\ref{optimalBW}) when $\alpha$ is $r^{\lceil \frac{k}{r} \rceil}$; 5. Access-optimality; 6. Fast decoding.
To the best of the authors' knowledge, these are the first code constructions for an arbitrary $\alpha$.
Motivated by the code construction in \cite{7084873}, we construct general codes where $\frac{k}{r}$ does not need to be an integer and $\alpha$ is not exclusively equal to $r^{\frac{k}{r}}$.
For instance, the code $(14, 10)$ that is deployed in the data-warehouse cluster of Facebook  \cite{Rashmi:2014:HGF:2619239.2626325} is out of the scope of applicability with the current proposals in \cite{7084873,6190343,6283042}, because $\frac{k}{r}=2.5$ is a non-integer.
However, the presented algorithm constructs an $(14, 10, 13)$ code that reduces the repair bandwidth for any systematic node by 67.5\% when $\alpha$ is $r^{\lceil \frac{k}{r} \rceil}=64$ compared to an $(14, 10)$ RS code.
The repair process is linear and highly parallelized. 
\vspace{-0.1cm}

\section{A General $(n, k, d=n-1)$ Code Construction} \label{general}

Consider a file of size $M = k \alpha$ symbols from a finite field $\mathbf{F}_q$ stored in $k$ systematic nodes $d_j$ of capacity $\alpha$ symbols.

We define a systematic MDS code in the following way: The basic data structure component is an index array of size $\alpha \times k$ where $\alpha \leq r^{\lceil \frac{k}{r} \rceil}$ and $n=k+r$, $P = ((i,j))_{\alpha \times k}$.
We use $r$ such index arrays $P_1,\ldots,P_r$.
The elements $p_{i,1}$, $i=1,\ldots,\alpha$, in $p_1$ are a linear combination only of the symbols with indexes present in the rows of $P_1$.
In the initialization phase, additional $\lceil \frac{k}{r} \rceil$ columns with pairs $(0, 0)$ are added to $P_2, \ldots, P_r$. 
The goal of the algorithm is to replace those zero pairs with concrete $(i, j)$ pairs so that the code is access-optimal for a given sub-packetization level $\alpha$. The value of $\alpha$ determines two phases of the algorithm. In the first phase, the indexes $(i,j)$ that replace the $(0, 0)$ pairs are chosen such that both Condition 1 and Condition 2 are satisfied (defined further in this section). The first phase starts with a granulation level parameter called $run$ that is initialized with the value $\lceil \frac{\alpha}{r}\rceil$. This parameter affects how the indexes $(i, j)$ are chosen and with every round the granulation level decreases by a factor $r$. Once the granulation level becomes equal to 1 and there are still $(0, 0)$ pairs that have to get some value $(i, j)$, the second phase starts where the remaining indexes are chosen such that only Condition 2 is satisfied. 

A high level description of the proposed algorithm is given in Alg. \ref{HighLvl}, while a detailed one in Alg. \ref{Constr}.
\begin{algorithm}
	\small
	\caption{High level description of an algorithm for generating general sub-packetized, access-optimal regenerating codes  
	}
	\label{HighLvl}
	\begin{algorithmic}[1]
		\State Initialize the index arrays $P_1, \ldots, P_r$;
		\State \# Phase 1
		\State Set the granulation level $run \leftarrow \lceil \frac{\alpha}{r}\rceil$
		\Repeat
			\State \parbox[t]{8cm}{Replace $(0, 0)$ pairs with indexes $(i, j)$ such that both Condition 1 and Condition 2 are satisfied;  \vspace{0.1cm}}
			\State Decrease the granulation level $run$ by a factor $r$.
		\Until{the granulation level $run > 1$}
		\State \# Phase 2
		\State If there are still $(0, 0)$ indexes that have to get some value $(i, j)$, choose them such that only Condition 2 is satisfied;
		\State Return the index arrays $P_1, \ldots, P_r$;
	\end{algorithmic}
\end{algorithm}

Once the index arrays $P_1, \ldots, P_r$ are determined, the symbols $p_{i,l}$ in the parity nodes, $1 \leq i \leq \alpha$ and $1 \leq l \leq r$, are generated as a combination of the elements $a_{j_1,j_2}$ where the pair $(j_1, j_2)$ is in the $i$-th row of the index array $P_l$, i.e.,
\begin{equation}\label{LinEquations}
p_{i,l}=\sum c_{l,i,j} a_{j_1,j_2}.
\end{equation}
The linear relations have to guarantee an MDS code, i.e., to guarantee that the entire information can be recovered from any $k$ nodes (systematic or parity).
We use the following terms and variables:

\begin{algorithm}
	\small
\caption{Algorithm to generate the index arrays
\newline
\textbf{Input:} $n, k, \alpha$;
\newline
\textbf{Output:} Index arrays $P_1, \ldots, P_r$.}
\label{Constr}
\begin{algorithmic}[1]
\State{\textbf{Initialization:} $P_1, \ldots, P_r$ are initialized as index arrays $P = ((i,j))_{\alpha \times k}$;}
\State{Append $\lceil \frac{k}{r}\rceil$ columns to $P_2, \ldots, P_r$ all initialized to $(0, 0)$;}
\State Set $portion \leftarrow \lceil \frac{\alpha}{r}\rceil$;
\State Set $ValidPartitions \leftarrow \emptyset$;
\State Set $j \leftarrow 0$;
\State \# Phase 1
\Repeat
	\State Set $j \leftarrow j+1$;
	\State Set $\nu \leftarrow \lceil \frac{j}{r}\rceil$;
	\State Set $run \leftarrow \lceil \frac{\alpha}{r^\nu}\rceil$;
	\State Set $step \leftarrow \lceil \frac{\alpha}{r}\rceil-run$;
	\State \parbox[t]{8cm}{$\mathcal{D}_{d_j} =$ $ValidPartitioning(ValidPartitions$, $k$, $r$, $portion$, $run$, $step$, $J_{\nu})$;  \vspace{0.1cm}}
	\State \parbox[t]{8cm}{Set $ValidPartitions=ValidPartitions \cup \mathcal{D}_{d_j}$; \vspace{0.1cm}}
	\State \parbox[t]{8cm}{Determine one $D_{\rho,d_j} \in \mathcal{D}_{d_j}$ such that its elements correspond to row indexes in the $(k+\nu)$-th column in one of the arrays $P_2, \ldots, P_r$, that are all zero pairs $(0, 0)$; \vspace{0.1cm}}
	\State \parbox[t]{8cm}{The indexes in $D_{\rho,d_j}$ are the row positions where the pairs $(i,j)$ with indexes $i \in \mathcal{D} \setminus D_{\rho,d_j}$ are assigned in the $(k+\nu)$-th column of $P_2, \ldots, P_r$; \vspace{0.1cm}}
\Until{$(run>1)\ \ \mbox{AND} \ \ (j \ne 0 \mod{r}) $}
\State \# Phase 2
\While{$j < k$}
	\State Set $j \leftarrow j+1$;
	\State Set $\nu \leftarrow \lceil \frac{j}{r}\rceil$;
	\State Set $run \leftarrow 0$;
	\State \parbox[t]{8cm}{$\mathcal{D}_{d_j} =$ $ValidPartitioning(ValidPartitions$, $k$, $r$, $portion$, $run$, $step$, $J_{\nu})$;  \vspace{0.1cm}}
	\State \parbox[t]{8cm}{Set $ValidPartitions=ValidPartitions \cup \mathcal{D}_{d_j}$; \vspace{0.1cm}}
	\State \parbox[t]{8cm}{Determine one $D_{\rho,d_j} \in \mathcal{D}_{d_j}$ such that its elements correspond to row indexes in the $(k+\nu)$-th column in one of the arrays $P_2, \ldots, P_r$, that are all zero pairs $(0, 0)$; \vspace{0.1cm}}
	\State \parbox[t]{8cm}{The indexes in $D_{\rho,d_j}$ are the row positions where the pairs $(i,j)$ with indexes $i \in \mathcal{D} \setminus D_{\rho,d_j}$ are assigned in the $(k+\nu)$-th column of $P_2, \ldots, P_r$; \vspace{0.1cm}}
\EndWhile
\State Return $P_1, \ldots, P_r$.
\end{algorithmic}
\end{algorithm}

\begin{itemize}
 \item The set $Nodes=\{d_1, \ldots, d_k\}$ of $k$ systematic nodes is partitioned in $\lceil \frac{k}{r}\rceil$ disjunctive subsets $J_1, J_2, \ldots, J_{\lceil \frac{k}{r}\rceil}$ where $|J_\nu|=r$ (if $r$ does not divide $k$ then $J_{\lceil \frac{k}{r}\rceil}$ has $k \mod{r}$ elements) and $Nodes=\cup_{\nu=1}^{\lceil \frac{k}{r}\rceil}J_\nu$.
In general, this partitioning can be any random permutation of $k$ nodes. Without loss of generality we use the natural ordering as follows:
$J_1=\{d_1,\ldots,d_r \} $, $J_2=\{d_{r+1},\ldots,d_{2r} \} $, $\ldots$ , $J_{\lceil \frac{k}{r}\rceil} = \{d_{ \lfloor \frac{k}{r} \rfloor \times r + 1},\ldots,d_{k} \} $. 	

	\item Each node $d_j$ consists of an indexed set of $\alpha$ symbols $\{a_{1,j},a_{2,j},\ldots,a_{\alpha,j}\}$.
	
	\item $portion = \lceil \frac{\alpha}{r}\rceil$: The set of all symbols in $d_j$ is partitioned in disjunctive subsets where at least one subset has $portion$ number of elements.
	\item $run=\lceil \frac{\alpha}{r^\nu}\rceil$, for values of $\nu \in \{1,\ldots,\lceil \frac{k}{r}\rceil\}$.
	\item $step=\lceil \frac{\alpha}{r}\rceil-run$: For the subsequent $(k+\nu)$-th column, where $\nu \in \{1,\ldots,\lceil \frac{k}{r}\rceil\}$, the scheduling of the indexes corresponding to the nodes in $J_\nu$ is done in subsets of indexes from a valid partitioning.
\item \emph{A valid partitioning} $\mathcal{D}_{d_j} = \{D_{1,d_j},\ldots,D_{r,d_j}\}$ of a set of indexes $D=\{1,\ldots,\alpha\}$, where the $i-$th symbol in $d_j$ is indexed by $i$ in $D$, is a partitioning in $r$ disjunctive subsets $D_{d_j} = \cup_{\rho=1}^{r}D_{\rho,d_j}$. If $r$ divides $\alpha$, then the valid partitioning for all nodes in $J_\nu$ is equal. If $r$ does not divide $\alpha$, then the valid partitioning has to contain at least one subset $D_{\rho,d_j}$ with $portion$ elements that correspond to the row indexes in the $(k+\nu)$-th column in one of the arrays $P_2, \ldots, P_r$ that are all zero pairs.
	




	
	\item \emph{Condition 1}: 
At least one subset $D_{\rho,d_j}$ has $portion$ elements with runs of $run$ consecutive elements separated with a distance between the indexes equal to $step$. The elements of that subset correspond to the row indexes in the $(k+\nu)$-th column in one of the arrays $P_2, \ldots, P_r$ that are all zero pairs. The distance between two elements in one node is computed in a cyclical manner such that the distance between the elements $a_{\alpha-1}$ and $a_2$ is 2.



	\item \emph{Condition 2}: A necessary condition for the valid partitioning to achieve the lowest possible repair bandwidth is $D_{d_{j_1}} = D_{d_{j_2}}$ for all $d_{j_1}$ and $d_{j_2}$ in $J_\nu$ and $D_{\rho,d_{j_1}}\neq D_{\rho,d_{j_2}}$ for all $d_{j_1}$ and $d_{j_2}$ systematic nodes in the system.
	If $portion$ divides $\alpha$, then $D_{\rho,d_j}$ for all $d_j$ in $J_\nu$ are disjunctive, i.e., $D = \cup_{j=1}^{r}D_{\rho,d_j}=\{1, \ldots, \alpha\}$.


\end{itemize}
\begin{algorithm}
	\small
	\caption{$ValidPartitioning$
		\newline
		\textbf{Input}:$ValidPartitions, k, r, portion, run, step, J_{\nu}$;
		\newline
		\textbf{Output}: $\mathcal{D}_{d_j} = \{D_{1,d_j},\ldots,D_{r,d_j}\}$.}
	\label{Valid}	
	\begin{algorithmic}[1]
		\State Set $D = \{1,\ldots,\alpha\}$;
		\If{$run \ne 0$}
		\State Find $\mathcal{D}_{d_j}$ that satisfies Condition 1 and Condition 2;
		\Else
		\State Find $\mathcal{D}_{d_j}$ that satisfies Condition 2;
		\EndIf
		\State Return $\mathcal{D}_{d_j}$;
	\end{algorithmic}
\end{algorithm}


Alg. \ref{AlgRepair} shows the repair of a systematic node where the systematic and the parity nodes are global variables.
A set of $\lceil \frac{\alpha}{r} \rceil$ symbols is accessed and transferred from each $n-1$ non-failed nodes. If $\alpha \ne r^{\lceil \frac{k}{r} \rceil}$, then additional elements may be required as described in Step 4.
Note that a specific element is transferred just once and stored in a buffer. For every subsequent use of that element, the element is read from the buffer and further transfer operation is not required.
The repair process is highly parallelized, because a set of $\lceil \frac{\alpha}{r} \rceil$ symbols is independently and in parallel repaired in Step 2 and then the remaining symbols are recovered in parallel in Step 5.
\vspace{-0.2cm}
\begin{algorithm}
	\small
\caption{Repair of a systematic node $d_l$
\newline
\textbf{Input:} $l$;
\newline
\textbf{Output:} $d_l$.}
\label{AlgRepair}
\begin{algorithmic}[1]
	\State Access and transfer $(k-1) \lceil \frac{\alpha}{r}\rceil$ elements $a_{i,j}$ from all $k-1$ non-failed systematic nodes and $\lceil \frac{\alpha}{r}\rceil$ elements $p_{i,1}$ from $p_1$ where $i \in D_{\rho,d_l}$;
	\State Repair $a_{i,l}$ where $i \in D_{\rho,d_l}$
	\State Access and transfer $(r-1)\lceil \frac{\alpha}{r}\rceil$ elements $p_{i,j}$ from $p_2, \ldots, p_{r}$ where $i \in D_{\rho,d_l}$;
	\State Access and transfer from the systematic nodes the elements $a_{i,j}$ listed in the $i-$th row of the arrays $P_2, \ldots, P_r$ where $i \in D_{\rho,d_l}$ that have not been read in Step 1;
		\State Repair $a_{i,l}$ where $i \in \mathcal{D} \setminus D_{\rho,d_l}$;
\end{algorithmic}
\end{algorithm}

\begin{proposition} \label{bw}
The repair bandwidth for a single systematic node $\gamma$ is bounded between the following lower and upper bounds:
\begin{equation}
\label{LowerUpperBounds}
\frac{n-1}{r} \leq \gamma \leq \frac{(n-1)}{\alpha} \lceil \frac{\alpha}{r} \rceil + \frac{(r-1)}{\alpha} \lceil \frac{\alpha}{r} \rceil
\lceil \frac{k}{r} \rceil.
\end{equation}
\end{proposition}

\begin{proof}
Note that we read in total $k \lceil \frac{\alpha}{r}\rceil$ elements in Step 1 of Alg. \ref{AlgRepair}. Additionally, $(r-1)\lceil \frac{\alpha}{r}\rceil$ elements are read in Step 3. Assuming that we do not read more elements in Step 4, we determine the lower bound as $k \lceil \frac{\alpha}{r}\rceil + (r-1)\lceil \frac{\alpha}{r}\rceil = (n-1)\lceil \frac{\alpha}{r}\rceil$ elements, i.e., the lower bound is $\frac{(n-1)}{\alpha} \lceil \frac{\alpha}{r}\rceil$ (since every element has a size of $\frac{1}{\alpha}$).
To derive the upper bound, we assume that we read all elements $a_{i,j}$ from the extra $\lceil \frac{k}{r}\rceil $ columns of the arrays $P_2, \ldots, P_r$ in Step 4. Thus, the upper bound is $\frac{(n-1)}{\alpha} \lceil \frac{\alpha}{r} \rceil + \frac{(r-1)}{\alpha} \lceil \frac{\alpha}{r} \rceil
\lceil \frac{k}{r} \rceil$.
\end{proof}


The optimality of the proposed code construction is captured in the following Proposition.
\begin{proposition} \label{key}
If $r$ divides $\alpha$, then the indexes $(i, j)$ of the elements $a_{i,j}$ where $i \in \mathcal{D} \setminus D_{\rho,d_j}$ for each group of $r$ systematic nodes are scheduled in one of the $\lceil \frac{k}{r} \rceil$ additional columns in the index arrays $P_2,\ldots,P_r$.
\end{proposition}
Next we show that there always exists a set of non-zero coefficients from $\mathbf{F}_q$ in the linear combinations so that the code is MDS.
We adapt Theorem 4.1 from \cite{7084873} as follows:
\begin{theorem}\label{MDS}
There exists a choice of non-zero coefficients $c_{l,i,j}$ where $l=1, \ldots, r$, $i=1, \ldots, \alpha$ and $j=1, \ldots, k$ from $\mathbf{F}_q$ such that the code is MDS if $q \geq \binom {n} {k} r \alpha$.
\end{theorem}

\begin{proof}
The system of linear equations in (\ref{LinEquations}) defines a system of $r \times \alpha$ linear equations with $k \times \alpha$ variables.
A repair of one failed node is given in Alg. \ref{AlgRepair}, but for the sake of this proof, we explain the repair by discussing the solutions of the system of equations. When one node has failed, we have an overdefined system of $r \times \alpha$ linear equations with $\alpha$ unknowns.
In general this can lead to a situation where there is no solution.
However, since the values in system (\ref{LinEquations}) are obtained from the values of the lost node, we know that there exists one solution. Thus, solving this system of $r \times \alpha$ linear equations with an overwhelming probability gives a unique solution, i.e., the lost node is recovered.
When 2 nodes have failed, we have a system of $r \times \alpha$ linear equations with $2\alpha$ unknowns.
The same discussion for the overdefined system applies here.
The most important case is when $r=n-k$ nodes have failed.
In this case, we have a system of $r \times \alpha$ linear equations with $r \times \alpha$ unknowns.
If the size of the finite field $\mathbf{F}_q$ is large enough, i.e., $q \geq \binom {n} {k} r \alpha$, as it is shown in Theorem 4.1 in \cite{7084873}, the system has a unique solution, i.e., the file $M$ can be collected from any $k$ nodes.
\end{proof}

\begin{figure}
\centering
\includegraphics[width=3.5in]{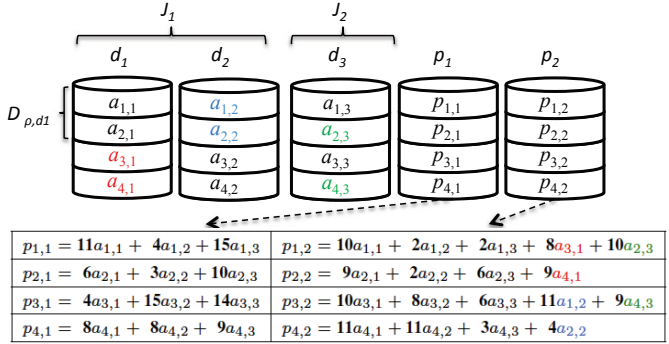}
\vspace{-0.5cm}
\caption{An MDS array code with 3 systematic and 2 parity nodes for $\alpha=4$. The elements presented in colors are scheduled as additional elements in $p_2$.The coefficients are from $\mathbf{F}_{16}$ with irreducible polynomial $x^4+x^3+1$. }
\label{five}
\end{figure}


\section{Code Examples} \label{example}

Let us take the $(5, 3, 4)$ code. We show a code construction for the optimal sub-packetization level $\alpha=2^{\lceil \frac{3}{2} \rceil}=4$. 

The following requirements have to be satisfied for the code to be an access-optimal MDS code that achieves the lower bound of the repair bandwidth for any systematic node:
\begin{itemize}
\item $M=k \alpha = 12$ symbols,
\item Repair a failed systematic node by accessing and transferring $\lceil \frac{\alpha}{r}\rceil=2$ symbols from the remaining $d=4$ nodes,
\item Reconstruct the data from any 3 nodes.
\end{itemize}

The systematic nodes $d_1, d_2, d_3$ and the parity nodes $p_1, p_2$ are shown in Fig. \ref{five}.
The file size is 12 symbols, where each node stores $\alpha=4$ symbols.
The elements of $p_1$ are linear combinations of the row elements from the systematic nodes multiplied by coefficients from $\mathbf{F}_{16}$.
The elements of $p_2$ are obtained by adding extra symbols to the row sum.
We next show the scheduling of an element $a_{i,j}$ from a specific $d_j$ where $i \in \mathcal{D} \setminus D_{\rho,d_j}$ at $portion=2$ positions in the $i$-th row, $i\in D_{\rho,d_j}$, and the $(3+\nu)$-th column, $\nu=1, 2$, of $P_2$.
We follow the steps in Alg. \ref{Constr} and give a brief description:
\newline
1. Initialize $P_1$ and $P_2$ as arrays $P = ((i,j))_{4 \times 3}$.
\newline
2. Append additional $2$ columns to $P_2$ initialized to $(0, 0)$.
\newline
3. Set $portion=2$ and $ValidPartitions=\emptyset$ .
\newline
4. For the nodes $d_1, d_2$ that belong to $J_1$, $run$ is equal to 2 and $step$ to 0. While $run$ is equal to 1 and $step$ to 1 for the node $d_3$ that belongs to $J_2$.
\newline
5. Alg. \ref{Valid} gives $D_{d_1}=\{\{1, 2\}, \{3, 4\}\}$.
Following step 14 in Alg. \ref{Constr}, the first 2 zero pairs in the $5$-th column of $P_2$ with $0$ distance between them are at the positions $(i,5)$ where $i=1,2$. Thus, $D_{\rho,d_1}=\{ 1, 2\}$.
We follow the same logic to obtain $D_{\rho,d_2}=\{3, 4\}$ and $D_{\rho,d_3}=\{1, 3\}$.
Next we schedule the elements of $d_j$ with $i$ indexes that are not elements of $D_{\rho,d_j}$ (represented in colors in Fig. \ref{five}) in the $i-$th row of $P_2$ where $i \in D_{\rho,d_j}$.
The $i-$th index of $a_{3,1}$ and $a_{4,1}$ does not belong to $D_{\rho,d_1}$ so these elements are scheduled at the positions $(1,5)$ and $(2,5)$ of $P_2$.
We add the elements $a_{1,2}, a_{2,2}$ from $d_2$ in the 3-rd and the 4-th row respectively, while we add the elements $a_{2,3}$ and $a_{4,3}$ from $d_3$ in the 1-st and 3-rd row respectively.
The symbols in the parity nodes are obtained as linear combinations of the row elements in the parity arrays. For instance, $p_{1,2}$ is a linear combination of the elements in the first row from all systematic nodes, $a_{3,1}$ and $a_{2,3}$.


We next show how to repair the node $d_1$ following Alg. \ref{AlgRepair}. First, we repair the elements $a_{1,1}, a_{2,1}$.
Thus, we access and transfer $a_{1,2}$, $a_{2,2}$, $a_{1,3}$ and $a_{2,3}$ from $d_2$ and $d_3$ and $p_{1,1}$, $p_{2,1}$ from $p_1$.
In order to recover $a_{3,1}, a_{4,1}$, we need to access and transfer $p_{1,2}$ and $p_{2,2}$ from $p_2$.
Hence, the data from $d_1$ is recovered by accessing and transferring in total 8 elements from 4 non-failed nodes.
Exactly the same amount of data, 8 symbols, is needed to repair $d_2$ or $d_3$.
Thus, the average repair bandwidth, defined as the ratio of the total repair bandwidth to repair all systematic nodes to the file size $M$, is equal to 2 symbols.
\subsection{Performance Analysis}
Another code discussed in this section is $(14, 10, 13)$ with different $\alpha$. Fig. \ref{14_10} shows the relation between the average repair bandwidth and $\alpha$.
For an RS code, $\alpha$ is 1 and the average repair bandwidth is equal to the file size.
A Hitchhiker code \cite{Rashmi:2014:HGF:2619239.2626325} for $\alpha=2$ reduces the repair bandwidth by 35$\%$ compared to the RS code. The remaining values of the average repair bandwidth are for the codes constructed with the algorithms presented in Section \ref{general}. We observe that the lower bound of the repair bandwidth that is 3.25 is achieved for $\alpha=r^{\lceil \frac{k}{r} \rceil}=64$. As we can see the repair bandwidth decreases as $\alpha$ increases.

\begin{figure}
\centering
\includegraphics[width=3.27in]{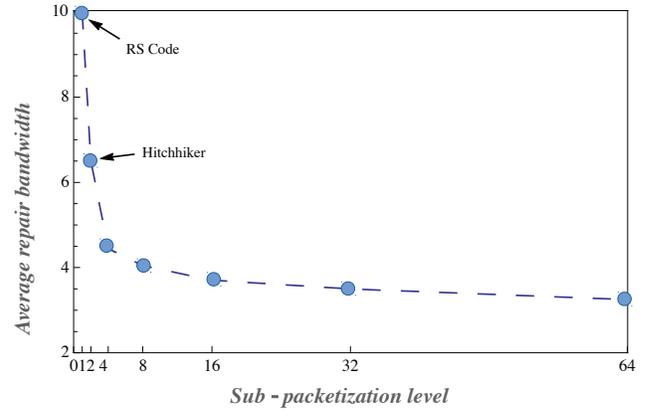}
\vspace{-0.3cm}
\caption{Average repair bandwidth for any systematic node for different sub-packetization levels $\alpha$ for an (14, 10, 13) code}
\label{14_10}
\end{figure}
\vspace{-0.1cm}
\section{Conclusions}
We presented a general construction of access-optimal regenerating codes that reach the lower bound of the repair bandwidth for $\alpha=r^{\lceil \frac{k}{r} \rceil}$, while the repair bandwidth is as close as possible to the lower bound when $\alpha < r^{\lceil \frac{k}{r} \rceil}$. The repair process of a systematic node is linear and highly parallelized.
\bibliographystyle{IEEEtran}
\bibliography{refer}

\end{document}